\documentclass[12pt]{article}
\usepackage[usenames,dvipsnames,svgnames,table]{xcolor} % usenames : 16개 기본색,  dvipsnames : 또 다른 68개 색, svgnames : 또 다른 150개 색에 접근할 수 있다. \usepackage{color} 대신 쓰인다. xcolor 패키지 설명서에서 색깔을 찾아볼 수 있다. 
\usepackage{kotex}
\usepackage{graphicx}
\usepackage{epsfig}
\usepackage{amssymb, amsmath, mathtools}
\usepackage{verbatim}
\usepackage{natbib}
\usepackage[affil-it]{authblk}
\usepackage{mathrsfs}
\usepackage{here}
\usepackage{makecell}
\usepackage{tikz}
\usepackage{enumerate}
\usepackage[shortlabels]{enumitem}
\usepackage{yfonts}
\usepackage{comment}
\usepackage{longtable}
\usepackage{tikz,amsmath, amssymb,bm,color}
\usetikzlibrary{circuits.logic.US,circuits.logic.IEC,fit}
\usepackage{algorithm2e}
\usepackage{multirow}

\usepackage{makecell}

\usepackage[colorlinks]{hyperref} % 목차와 내용을 링크해 준다. todonotes 패키지에 필요하다. 목차가 빨간색이 된다. colorlinks 옵션을 쓰지 않으면 검은색이 된다. 
\usepackage[colorinlistoftodos, textsize=scriptsize]{todonotes} % 노트를 만들기 위한 패키지. disable 옵션을 쓰면 노트를 모두 없애준다. 

\usepackage[framemethod=TikZ]{mdframed}
\mdfdefinestyle{JL}{%
	linecolor=blue,
	outerlinewidth=1pt,
	roundcorner=2pt,
	innertopmargin=0.1\baselineskip,
	innerbottommargin=0.1\baselineskip,
	innerrightmargin=2pt,
	innerleftmargin=2pt,
	backgroundcolor=orange!100!white}

\hypersetup{
	colorlinks,
	citecolor=Black,
	linkcolor=Red,
	urlcolor=Blue}

\includecomment{note} %\begin{note} \end{note} 혹은 \note \endnote로 된 부분을 포함한다.
\newtheorem{theorem}{Theorem}[section]
\newtheorem{lemma}[theorem]{Lemma}
\newtheorem{definition}[theorem]{Definition}
\newtheorem{proposition}[theorem]{Proposition}
\newtheorem{corollary}[theorem]{Corollary}

\newtheorem{remark}[theorem]{Remark}

\newenvironment{proof}[1][Proof]{\begin{trivlist}
		\item[\hskip \labelsep {\bfseries #1}]}{\end{trivlist}}

\everymath{\displaystyle}

\newcommand{\se}[1]{\section{#1}}
\newcommand{\sse}[1]{\subsection{#1}}

\newcommand{\be}{\begin{equation}}
	\newcommand{\ee}{\end{equation}}
\newcommand{\bea}{\begin{eqnarray*}}
	\newcommand{\eea}{\end{eqnarray*}}
\newcommand{\bean}{\begin{eqnarray}}
	\newcommand{\eean}{\end{eqnarray}}
\newcommand{\ben}{\begin{enumerate}}
	\newcommand{\een}{\end{enumerate}}
\newcommand{\bi}{\begin{itemize}}
	\newcommand{\ei}{\end{itemize}}
\newcommand{\brem}{\begin{remark}}
	\newcommand{\erem}{\end{remark}}
\newcommand{\bcen}{\begin{center}}
	\newcommand{\ecen}{\end{center}}
\newcommand{\bsv}{\begin{semiverbatim}}
	\newcommand{\esv}{\end{semiverbatim}}

\newcommand{\bt}{\begin{theorem}}
	\newcommand{\et}{\end{theorem}}
\newcommand{\bl}{\begin{lemma}}
	\newcommand{\el}{\end{lemma}}
\newcommand{\bd}{\begin{definition}}
	\newcommand{\ed}{\end{definition}}
\newcommand{\bc}{\begin{corollary}}
	\newcommand{\ec}{\end{corollary}}
\newcommand{\bp}{\begin{proposition}}
	\newcommand{\ep}{\end{proposition}}

\newcommand{\bbN}{ \mathbb{N}}

\usepackage{xr}
\makeatletter
\newcommand*{\addFileDependency}[1]{
	\typeout{(#1)}
	\@addtofilelist{#1}
	\IfFileExists{#1}{}{\typeout{No file #1.}}
}
\makeatother

\title{Seroprevalence of SARS-CoV-2 antibodies in South Korea}
\author[1]{Kwangmin Lee}
\author[2]{Seongil Jo}
\author[1]{Jaeyong Lee\footnote{Corresponding author: leejyc@gmail.com}}
\affil[1]{Department of Statistics, Seoul National University}
\affil[2]{Department of Statistics, Inha University}

\begin{document}

	\maketitle
	%\setcounter{secnumdepth}{1}
	%\setcounter{tocdepth}{2}
	%\tableofcontents

	\begin{abstract}
		In $2020$, Korea Disease Control and Prevention Agency reported three rounds of surveys on seroprevalence of severe acute respiratory syndrome coronavirus 2 (SARS-CoV-2) antibodies in South Korea. We analyze the seroprevalence surveys using a Bayesian method with an informative prior distribution on the seroprevalence parameter, and  the sensitivity and specificity of the diagnostic test.  We construct the informative prior  using the posterior distribution obtained from the clinical evaluation data based on the plaque reduction neutralization test. 
The constraint of the seroprevalence parameter induced from the known confirmed  cornonavirus 2019 cases can be imposed naturally  in the proposed Bayesian model. 
We also prove that the confidence interval of the seroprevalence parameter based on the Rao's test can be the empty set, while the Bayesian method renders a reasonable interval estimator. 
        As of the $30$th of  October 2020,  the $95\%$ credible interval of the estimated SARS-CoV-2 positive population does not exceed $307,448$, approximately  $0.6\%$ of the Korean population.
	\end{abstract}
	
	\se{Introduction}
	
	In December 2019, the Chinese government reported a cluster of pneumonia patients of unknown cause in Wuhan, China. It was found that an unknown betacoronavirus causes the disease \citep{zhu2020novel}.
	The Coronaviridae Study Group (CSG) of the International Committee on Taxonomy of Viruses has named the virus as severe acute respiratory syndrome coronavirus 2 (SARS-CoV-2), due to the similarity to SARS-CoV \citep{gorbalenya2020species}.
	The World Health Organization (WHO) also has named the disease caused by SARS-CoV-2 as COVID-19, short for coronavirus disease 2019 \citep{world2020novel22}.
	As of  January 10, 2021, over $90,000,000$ people in the world are confirmed positive for COVID-19, and there are over $68,000$ confirmed cases in South Korea.

    Most statistical approaches use the number of confirmed cases to assess the spread of infectious diseases in a population. However, the number of confirmed cases does not include those that are infected but not detected. 
    A seroprevalence survey can be an alternative in this case. The seroprevalence is the number of people with antibodies to the virus in a population.
     The WHO (2020b)  \nocite{world2020population} proposes to analyze seroprevalence  surveys for the inference on the spread of  a novel coronavirus.
%    Although the number of confirmed cases is widely used to assess the spread of infectious diseases in a population, this is a dubious measure for the purpose.
%    The number of confirmed cases can underestimate the spreading scale since there might be cases infected but not detected by the government. 
%Furthermore, a comparison study based on this measure can be distorted since the underestimation scales can be different between populations. 
%    In this situation, the seroprevalence survey for infectious diseases can be an alternative for the purpose.
%    The seroprevalence is the number of people with antibodies to the virus in a population, and it allows inferences on the spread of infection for a novel coronavirus \citep{world2020population}.
%    Contrary to the number of confirmed cases, the seroprevalence survey is based on survey sampling, which allows unbiased estimation.
    Seroprevalence surveys have been conducted in many countries, and the results are collected in Serotracker, a global seroprevalence dashboard \citep{arora2020serotracker}.
    According to the recent update on December $12$, $2020$, Serotracker provides the survey results of $56$ countries based on $491$ studies.

%There are three approaches to statistical analysis for the seroprevalence survey.

The seroprevalence survey data can be analyzed under either the assumption that  the diagnostic tests used in the survey are $100\%$ accurate or the assumption that the tests are not $100\%$ accurate. We will term these assumptions as the {\it accuracy assumption} and the {\it inaccuracy assumption}, respectively. Under the accuracy assumption,  \cite{song2020igg} and \cite{noh2020seroprevalence} analyzed outpatient data sets  in southwestern Seoul and Daegu, respectively, and estimated the seroprevalence. 
Although the statistical models are simpler under the accuracy assumption, the estimates   can be biased unless  the assumption is met as pointed out in \cite{diggle2011estimating}.  Under the inaccuracy assumption, 
\cite{diggle2011estimating} proposed a corrected prevalence estimator  and \cite{silveira2020population} constructed a confidence interval of the seroprevalence  using a resampling method. In an analysis of a seroprevalence survey data  of southern Brazil, \cite{silveira2020population}  showed that confidence intervals can be $\{ 0 \} $, which is hardly reliable.  See Supplementary Table 2 in \cite{silveira2020population}.
In Section \ref{sec:frequentist},  we also prove that the confidence interval constructed from the Rao's test using the duality theorem \citep{bickel2015mathematical} can be the empty set. These examples show that the frequentist confidence intervals of the seroprevalence under the inaccuracy assumption can be unreliable.

%The last is the Bayesian approach. \cite{stringhini2020seroprevalence} used this approach to analyze the seroprevalence of SARS-CoV-2 in Geneva, Switzerland. 

%Using the Bayesian method, we can reflect prior information, which is specified in Section \ref{sse:statisticalanalysis}.
%The Bayesian method has another advantage over the frequentist method since the frequentist method may yield useless confidence intervals.
%For example, \cite{silveira2020population} analyzed the survey in Southern Brazil and showed that confidence intervals are $[0,0]$, which is hard to rely on. See Supplementary Table 2 in \cite{silveira2020population}.
%We analyze the issue by constructing the confidence interval explicitly by duality theorem, covered in detail in Section \ref{sse:statisticalanalysis}.

In this paper, we propose a Bayesian method under the inaccuracy assumption 
and apply the proposed method to  the  seroprevalence surveys of the South Korean population conducted  in $2020$  \citep{kcdc}. 
We use the posterior distribution obtained from the Bayesian model of the clinical evaluation data \citep{kohmer2020brief} as the informative prior distribution of the sensitivity and specificity on the diagnostic test.

The rest of the paper is organized as follows. In the next section, we describe the seroprevalence surveys of SARS-CoV-2 motivating this work and the plaque reduction neutralization test for detection of SARS-CoV-2 antibodies. 
In Section \ref{sec:frequentist}, we conduct a frequentist analysis and discuss the phenomenon of empty confidence sets. In Section \ref{sec:Bayesian}, we propose a Bayesian method for the seroepidemiological survey that gives nonempty interval estimates, and analyze the seroprevalence surveys of the South Korean population using the proposed Bayesian method. We conclude the paper with a discussion section.

%%%%%%%%%%%%%%%%%%%%%%%%%%%%%%%%%%%%%%%%%%%%%%%%%%%%%%%%%%%%%%%%%%%%%%%%%%%
	\se{Seroepidemiological surveys and clinical evaluation of a serology test}
%%%%%%%%%%%%%%%%%
%%%%%%%%%%%%%%%%%
	\sse{Seroepidemiological surveys of SARS-CoV-2 in South Korea}\label{sec:survey}

	Korea Disease Control and Prevention Agency (KDCA) conducted three rounds of seroprevalence surveys of SARS-CoV-2 for South Korean population in 2020. KDCA used the sets of samples collected in the Korea National Health and Nutrition Examination Survey (KNHNES), which is a regular national survey to investigate the health and nutritional status of South Koreans since 1998 \citep{kweon14}, as the samples of the seroprevalance surveys.
	KDCA performed a serology test for SARS-CoV-2 to the residual serums, and the test results \citep{kcdc} are summarised in Table \ref{tbl:survey}. 
	In Table \ref{tbl:survey}, the periods during which the samples are collected are also given. 
	\begin{table}[!ht]
		\caption{ The result of the seroprevalence surveys in 2020 \citep{kcdc}. The column of the announcement date represents dates when KDCA reports the results of the surveys. The column of the collection period represents the periods during which the sets of samples are collected.\label{tbl:survey}}
		\centering
		\begin{tabular}{|c|c|c|c|}
			\hline
			Accouncement date & Collection period  & Number of samples & \makecell{Number of \\  test-positive samples} \\ \hline
			9th of July            & 4.21. $\sim$ 6.16.  & $1500$              & $0$                       \\ \hline
			11th of September      & 6.10. $\sim$ 8.13.  & $1440$              & $1$                       \\ \hline
			23th of November       & 8.14. $\sim$ 10.31. & $1379$              & $3$                       \\ \hline
		\end{tabular}  
	\end{table}

%	KNHNES is a regular national survey to investigate national health irrespective of epidemics of COVID-19, and the population of the survey sample is the whole South Korean population.
%	Thus when we consider the survey population, this includes the confirmed patients, and the sample may contain the confirmed cases. For example, the sample in the third round of the survey includes two confirmed cases.

%%%%%%%%%%%%%%%%%
%%%%%%%%%%%%%%%%%
	\sse{Clinical evaluation of plaque reduction neutralization test for SARS-CoV-2 antibodies}
	
	When KDCA performed a serology test for SARS-CoV-2, KDCA used their in-house plaque reduction neutralization test (PRNT). In the PRNT, serum samples are tested for their neutralization capacity against SARS-CoV-2. To estimate the sensitivity and specificity of PRNT methods for SARS-CoV-2 empirically, we use a set of clinical evaluation data (Table \ref{table:frank}) which is conducted by \cite{kohmer2020brief}.
%	While KCDC does not open to the public any information on the sensitivity and specificity of their in-house method, \cite{kohmer2020brief} presents the experimental result of clinical evaluation for their methods.
%	The experiment result is summarised in Table \ref{table:frank}.
	\begin{table}[!htbp]
		\centering
			\caption{The data of clinical evaluation of the PRNT by \cite{kohmer2020brief}.
			The columns represent the true states of samples. The true state of a sample refers to whether the sample has the antibodies against SARS-CoV-2 in reality.
			\label{table:frank}}
		\begin{tabular}{|c|c|c|c|c|}
			\hline
			\multicolumn{2}{|c|}{\multirow{2}{*}{}}   & \multicolumn{3}{c|}{True state}          \\ \cline{3-5} 
			\multicolumn{2}{|c|}{}                    & Positive     & Negative     & Total            \\ \hline
			\multirow{3}{*}{\makecell{Test results \\ of the PRNT}} & Positive & $42$     & $1$     & $43$     \\ \cline{2-5} 
			& Negative & $3$     & $34$     & $37$     \\ \cline{2-5} 
			& Total    & $45$ & $35$ & $80$ \\ \hline
		\end{tabular}

	\end{table}

%%%%%%%%%%%%%%%%%
%%%%%%%%%%%%%%%%%
	\se{Maximum likelihood estimator and a confidence interval}\label{sec:frequentist}

	%We present a Bayesian approach to estimate the percentage of people in a population who have antibodies to SARS-CoV-2, $p^{(p)}$. To do this, we first show that frequentist approaches have limitations, and then we try to tackle them by introducing a Bayesian method, which utilizes a constrained prior distribution for the seroprevalence parameter.

	Under the inaccuracy assumption, we specify a statistical model for seroprevalance surveys, and present the maximum likelihood estimator and a confidence interval of the seroprevalance. We assume that the sensitivity and specificity of serology test are fixed values for the estimator and the confidence interval. 
	Note that the sensitivity and specificity are the probabilities that the positive has the positive test result and the negative has the negative test result, respectively. 
	
	We define {\it seroprevalance parameter}, $\theta$, as the proportion of those who have antibodies against SARS-CoV-2 in the population. 
	Let $ N $ be the number of samples of seroprevalence survey, $ X $ be the number of test-positive samples by serology test, and $p_+$ and $p_-$ denote the sensitivity and specificity of the serology test, respectively. We assume $X$ is generated from the binomial distribution:
%	The sensitivity and specificity of the test are denoted by $p_+^{(1)}$ and $p_-^{(1)}$, respectively. 
%	If cases in the survey are sampled randomly, the data are generated by the model below:
	\bean
	X \sim  Binom\left(N,\theta p_+ +(1-\theta)(1-p_-)\right),\label{model:gen1}
	\eean
	where $Binom(n,p)$ denotes the binomial distribution with parameters $n \in \mathbb{N}$ and $p \in [0, 1]$.
	When $p_+$ and $p_-$ are known, the maximum likelihood estimator for $\theta$ is as follows. 
	If $1-p_-<p_+$, then 
	\bean\label{eq:mle}
	\hat{\theta}^{MLE} =  \begin{cases} 0 & \text{if } X \le N(1-p_-) \\
	1 &\text{if } X \ge Np_+\\
	\frac{X/N - (1-p_-)}{p_+ +p_- -1} & \text{if } N(1-p_-)< X < N p_+,
	\end{cases}  
	\eean
	and if $p_+ < 1-p_-$, then
	\bea
	\hat{\theta}^{MLE} =  \begin{cases} 0 & \text{if } X \ge N(1-p_-) \\
		1 &\text{if } X \le Np_+\\
		\frac{X/N - (1-p_-)}{p_+ +p_- -1} & \text{if } Np_+ <  X< N(1-p_-). 
	\end{cases}  
	\eea
	Note if the number of test-positive samples is small or large enough, the maximum likelihood estimator can be $0$ or $1$. This means that nobody or everybody in the population has antibodies against SARS-CoV-2, which is hardly reliable. 
	
	We construct a confidence interval of $\theta$ from Rao's test \citep{rao1948large} using the duality thoerem \citep{bickel2015mathematical}, and show that when $X$ is too small or large, the confidence interval can be the empty set. 
	Let $A(\theta_0) =[l_{\theta_0}, u_{\theta_0}]$ be the $100(1-\alpha)\%$ acceptance interval of the Rao's test under the null hypothesis $H_0: \theta=\theta_0$. By the duality theorem $S(X)= \{\theta_0\in[0,1]: X \in A(\theta_0) \}$ is a $100(1-\alpha)\%$ confidence interval for $\theta$.
	Theorem \ref{theorem:raotest} gives the acceptance interval, $A(\theta_0)$, and the condition that the confidence interval $S(X)$ is the empty set.  
	
	\begin{theorem}\label{theorem:raotest} Consider the model \eqref{model:gen1}. 
	\begin{enumerate}[(a)]
	\item 
    The $100(1-\alpha)\%$ acceptance region of the test
      $$H_0 : \theta = \theta_0 \text{ vs } H_1 : \theta \neq \theta_0$$
     by the Rao’s Score test is given as
    \bea
    [l_{\theta_0}, u_{\theta_0} ] = [ N\theta^*_0 - \{N\chi^2_{0.05}(1)\theta^*_0(1-\theta^*_0) \}^{1/2}, N\theta^*_0 + \{N\chi^2_{0.05}(1)\theta^*_0(1-\theta^*_0) \}^{1/2} ],
    \eea
    where $\theta^*_0 = \theta_0p_+ + (1-\theta_0)(1-p_-)$ and $\chi^2_{\alpha}(1)$ is $(1-\alpha)100\%$ quantile of chi-square distribution with $1$ degree of freedom.
    \item If
	\bean
	X<\inf_{\theta_0\in[0,1]}l_{\theta_0} \text{ or } X>\sup_{\theta_0\in[0,1]}u_{\theta_0},\label{formula:freq}
	\eean
	the $100(1-\alpha)\%$ confidence interval $S(X)= \{\theta_0\in[0,1]: X \in A(\theta_0) \}$ is the empty set. 
    \end{enumerate}
	\end{theorem}
	\begin{proof} (a)
	Let $\theta^* = \theta p_+ + (1-\theta)(1-p_-)$ and
	$$L(\theta)= L(\theta;N,X) = \binom{N}{X} \theta^X (1-\theta)^{N-X}.$$
	The score statistics is 
	\bea
	\Big(\frac{d\log L(\theta)}{d\theta} \Big)_{\theta=\theta_0}^2 \Big[  E\Big(-\frac{d^2\log L(\theta)}{d\theta^2}  \Big)_{\theta=\theta_0}\Big]^{-1}
	&=& \Big(\frac{d\log L(\theta)}{d\theta^*} \Big)_{\theta=\theta_0}^2 \Big[  E\Big(-\frac{d^2\log L(\theta)}{d(\theta^*)^2}  \Big)_{\theta=\theta_0}\Big]^{-1}\\
	&=&\frac{(X-N\theta_0^*)^2}{N\theta_0^*} + \frac{(N-X-N(1-\theta_0^*))^2}{N(1-\theta_0^*)}\\
	&=& \frac{(X-N\theta_0^*)^2}{N\theta_0^*(1-\theta_0^*)}.
	\eea
	Then,  $100(1-\alpha)\%$ acceptance interval is 
	\bea
	A(\theta_0) &=& \Big\{ X : \frac{(X-N\theta_0^*)^2}{N \theta_0^*(1-\theta_0^*)} \le\chi^2_{\alpha}(1)  \Big\}\\
	&=& \{ X : |X-N\theta_0^*| \le  \{\chi^2_{\alpha}(1)N p_0^*(1-\theta_0^*)\}^{1/2} \},
	\eea
	which proves (a). 
	
\noindent 	(b) If
	\bea
	X<\inf_{\theta_0\in[0,1]}l_{\theta_0} \text{ or } X>\sup_{\theta_0\in[0,1]}u_{\theta_0},
	\eea
	then $X \notin A(\theta_0)$ for all $\theta_0\in[0,1]$. It implies the confidence interval of $X$ is the empty set. This completes the proof. $\blacksquare$ 
	\end{proof}
	
The intuitive reason for the empty confidence interval is as follows. 
The set of sampling distributions for $X$ is 
\bea
\left\{Binom(N,\theta): (1-p_- ) \le \theta \le p_+   \right\}.
\eea
 When $X/N$ is  smaller (larger) than  $1-p_-$ ($p_+$), the probability that $X$ is observed is small for every sampling distribution in the set. This makes test decisions rejected for every null hypothesis. Thus, the extreme $X$ implies $p_-$ and $p_+$ are doubtful.

    For the three rounds of surveys given in Table \ref{tbl:survey}, we show all the maximum likelihood estimators are zero and the confidence intervals are the empty set. We assume the fixed $(p_+,p_-)$ to be $(42/45,34/35)$, which is calculated from the clinical evaluation data (Table \ref{table:frank}) and formula $(r_{++}/r_{\cdot +}, r_{--}/r_{\cdot -})$ according to the notation in Table \ref{table:tvstrue}. 
    Based on equation \eqref{eq:mle}, all the maximum likelihood estimatiors are zero, since values of $N(1-p_-)$ are $42.9$, $41.1$ and $39.4$ which are all larger than the observed $X$s. 
    The confidence intervals are the empty set since values of $\inf_{\theta\in[0,1]} l_{\theta}$ are $30.2$, $28.8$ and $27.3$ which satisfy condition \eqref{formula:freq} in Theorem \ref{theorem:raotest}.

\se{A Bayesian method with informative prior distributions}\label{sec:Bayesian}

We propose a Bayesian method that avoids the empty confidence set problem. 
For the Bayesian analysis of model \eqref{model:gen1}, we assign prior distributions on $\theta$, $p_+$ and $p_-$. 
According to KNHNES design, the parameter $\theta$ refers to the seroprevalence in the population that includes those who have been confirmed to be tested positive for COVID-19 by the government. Thus, it is reasonable to assume $\theta$ is larger than the proportion of the confirmed cases, and we choose the following constrained prior distribution on parameter $\theta$:
\bean
\pi(\theta) \propto   (\theta)^{-1/2} (1-\theta)^{-1/2}  I(\theta>\tilde{\theta}),\label{eq:prior}
\eean
where $\pi(\theta)$ is the density function of the prior distribution on $\theta$, and  $\tilde{\theta}$ is the total number of confirmed cases divided by the number of the population. Note that the constrained prior distribution \eqref{eq:prior} is constructed by constraining Jefferey's prior or reference prior distribution for binomial parameter \citep{yang1996catalog}.

To construct prior distributions on $p_+$ and $p_-$, we use the posterior distribution on the sensitivity and specificity obtained from a clinical evaluation of the serology test.
In the clinical evaluation, we consider that the serology test is applied to samples of which the true states are known. The true state of a sample refers to whether the sample has the antibodies in reality. The data from the clinical evaluation is then represented as Table \ref{table:tvstrue}.
\begin{table}[!htbp]
	\centering
	\caption{Data format for clinical evaluation when the true states of samples are known.\label{table:tvstrue}}
	\begin{tabular}{|c|c|c|c|c|}
		\hline
		\multicolumn{2}{|c|}{\multirow{2}{*}{}}   & \multicolumn{3}{c|}{True state}          \\ \cline{3-5} 
		\multicolumn{2}{|c|}{}                    & Positive     & Negative     & Total            \\ \hline
		\multirow{3}{*}{
			Test result} & Positive & $r_{++}$     & $r_{+-}$     & $r_{+\cdot}$     \\ \cline{2-5} 
		& Negative & $r_{-+}$     & $r_{--}$     & $r_{+\cdot}$     \\ \cline{2-5} 
		& Total    & $r_{\cdot+}$ & $r_{\cdot-}$ & $r_{\cdot\cdot}$ \\ \hline
	\end{tabular}
	
\end{table}

For the analysis of the clinical evaluation (Table \ref{table:tvstrue}), we specify a statistical model using the binomial distribution as 
\bea
r_{++} &\sim& Binom(r_{\cdot +}, p_+)\\
r_{--} &\sim& Binom(r_{\cdot -}, p_-).
\eea
By applying the Jefferey's prior (or reference prior) to the binomial parameters $p_+$ and $p_-$, we obtain the densitiy functions of posterior distributions, $\pi^*(p_+ \mid r_{++}, r_{\cdot +} )$ and $\pi^*(p_- \mid r_{--}, r_{\cdot -} )$, as 
\bean\label{model:gen2}
\pi^*(p_+ \mid r_{++}, r_{\cdot +} ) &\propto&  p^{Binom}(r_{++} \mid r_{\cdot +},p_+) (p_+)^{1/2} (1-p_+)^{1/2}\nonumber\\
\pi^*(p_- \mid r_{--}, r_{\cdot -} ) &\propto&  p^{Binom}(r_{--} \mid r_{\cdot -},p_-) (p_-)^{1/2} (1-p_-)^{1/2},
\eean
where $p^{Binom}(\cdot \mid n,p)$ is the densitiy function of the binomial distribution $Binom(n,p)$ for $n\in\bbN$ and $p\in[0,1]$. 
Note that the Jefferey's prior is a conjugate prior for the likelihood function $p^{Binom}(\cdot \mid n,p)$. Thus, the density function of the posterior distributions are calculated as 
\bea
\pi^*(p_+ \mid r_{++}, r_{\cdot +} ) &\propto&  (p_+)^{(r_{++}+1/2)} (1-p_+)^{( r_{\cdot +} -r_{++}+ 1/2)}\\
\pi^*(p_+ \mid r_{--}, r_{\cdot -} ) &\propto& (p_-)^{( r_{--}+1/2)} (1-p_-)^{(r_{\cdot -}- r_{--}+1/2)}.
\eea
Finally, we use the posterior distributions to construct the informative prior distributions on $p_+,p_-$ of model $\eqref{model:gen1}$. That is, we set 
\bea
\pi(p_+)  &\propto&  (p_+)^{(r_{++}+1/2)} (1-p_+)^{( r_{\cdot +} -r_{++}+ 1/2)}\\
\pi(p_-)  &\propto& (p_-)^{( r_{--}+1/2)} (1-p_-)^{(r_{\cdot -}- r_{--}+1/2)},
\eea
where $\pi(p_+)$ and $\pi(p_-)$ are the densitiy functions of the informative prior distributions.  
%구체적으로 사후분포가 어떻게 계산되는지 나타낼것.

    We analyze the survey data (Table \ref{tbl:survey}) using the proposed Bayesian method. 
    Let $\theta_1$, $\theta_2$ and $\theta_3$ be the seroprevalance parameters for each survey.
    We assign the constrained prior distributions on $\theta_1$, $\theta_2$ and $\theta_3$ as equation \eqref{eq:prior}.
    When calculating the percentage of the confirmed cases, we use the cumulative confirmed cases at the last dates in the collection periods of the sets of samples, and let them denoted by $\tilde{\theta}_1$, $\tilde{\theta}_2$ and $\tilde{\theta}_3$.
    We construct informative prior distributions on $p_+$ and $p_-$ using the clinical evaluation of the PRNT for SARS-CoV-2 performed by \cite{kohmer2020brief}.
    By applying the clinical evaluation data (Table \ref{table:frank}) to equation \eqref{model:gen2}, we obtain the informative prior distributions as 
    \bea
    p_+ &\sim& Beta(42.5,3.5)\\
	p_- &\sim& Beta(34.5,1.5),
    \eea
    where $Beta(\alpha,\beta)$ denotes the beta distribution with the density function of $$f(x) = \frac{x^{\alpha-1}(1-x)^{\beta-1}}  {\int_0^1 x^{\alpha-1}(1-x)^{\beta-1} dx}.$$
    Collecting the prior distributions and three rounds of seroprevalance survey results, we construct the generative model as 
	\bea
	X_1 \mid \theta_1, p_+, p_-  &\sim& Binom(N_1, \theta_1 p_+ + (1-\theta_1)(1-p_-))\\
	X_2  \mid \theta_2, p_+, p_- &\sim& Binom(N_2, \theta_2 p_+ + (1-\theta_2)(1-p_-))\\
	X_3 \mid \theta_3, p_+, p_-  &\sim& Binom(N_3, \theta_3 p_+ + (1-\theta_3)(1-p_-))\\
	p_+ &\sim& Beta(42.5,3.5)\\
	p_- &\sim& Beta(34.5,1.5)\\
	\pi(\theta_1) &\propto& (\theta_1)^{-1/2}(1-\theta_1)^{-1/2}  I(\theta_1 \ge \tilde{\theta}_1 )\\
	\pi(\theta_2) &\propto& (\theta_2)^{-1/2}(1-\theta_2)^{-1/2}I(\theta_2 \ge \tilde{\theta}_2 )\\
	\pi(\theta_3) &\propto& (\theta_3)^{-1/2}(1-\theta_3)^{-1/2}I(\theta_3 \ge \tilde{\theta}_3 ),
	\eea
	where $(N_i,X_i)$ is the pair of the number of samples and the number of test-positive samples of $i$th seroprevalance survey for $i\in \{1,2,3\}$. 
	
	For inference, we generate posterior samples using Markov chain Monte Carlo (MCMC) sampling method. Specifically, we generate 4,000 posterior samples through running 4 Markov chains with different initial values, where each chain has 1,000 samples after a burn-in period of 1,000 samples.
%	we use Stan \citep{carpenter2017stan} to sample the posterior sample of $p^{(p)}_1$, $p^{(p)}_2$ and $p^{(p)}_3$. 
	We implement the MCMC algorithm with Stan \citep{carpenter2017stan}. 
	We extract the posterior samples of $\theta_1$, $\theta_2$ and $\theta_3$, and multiply the number of the population in 2020, $51,829,023$ \citep{MOIS}, to the parameters.
	We then give the summary statistics of the multiplied posterior samples in Table \ref{table:summarypos}.
		\begin{table}[!htbp]
		\centering
		\caption{Summary statistics of posterior distributions of the population who has antibodies against SARS-CoV-2 for the three rounds of the seroprevalance surveys. 
			The date column represents the last dates of the collection period of each survey. The column of cumulative confirmed cases represents the cumulative numbers of confirmed cases on the corresponding dates. 
			\label{table:summarypos} }
		\begin{tabular}{|c|c|c|c|}
			\hline
			Date     &\makecell{Cumulative \\  confirmed cases }    & Posterior mean & The 95\% credible interval   \\ \hline
			16th of June & $12198$ & $38380.0$       & $[12544.8, 122919.7]$      \\ \hline
			13th of July &  $14873$ & $58742.9$        & $[16126.5, 175693.4]$      \\ \hline
			31th of October & $26635$  & $119979.7$       & $[31159.9, 307448.1]$       \\ \hline
		\end{tabular}
		
	\end{table}
	According to Table \ref{table:summarypos}, the ratio of the posterior mean to the confirmed cases ranges from $3.1$ to $4.5$, which represents the proportion of $$\frac{[\text{The total number of the infected}] }{ [\text{The total number of the detected}]}.$$

	Finally, we compare the result of the proposed Bayesian method with the cumulative number of confirmed cases and the result of statistical analysis under the accuracy assumption. 
    Under the accuracy assumption, we consider the statistical model 
    \bea
    X \sim  Binom\left(N,\theta\right),
    \eea
    instead of model \eqref{model:gen1}. We use $X/N$ as a point estimator for $\theta$, and we construct a confidence interval of $\theta$ by \cite{clopper1934use}. 
    As in the proposed Bayesian method, we multipy the number of the population to the point estimator and the confidence interval.
    The comparison is then represented in Figure \ref{fig:sero}.
     \begin{figure}[!htbp]
   	\centering
   	\caption{
   		Dots and error bars denoted by ``Bayesian method" represent the posterior mean and $95\%$ credible intervals of the multiplied posterior distributions on seroprevalance parameters by the proposed Bayesian method.
   		Dots and error bars denoted by ``Accuracy assumption" represent the multiplied point estimators and the multiplied $95\%$ confidence intervals of the results of statistical analysis under the accuracy assumption. 
   		The line graph denoted by ``Confirmed" represents the daily cumulative confirmed cases.}
   	\includegraphics[height=10.5cm,width=15cm]{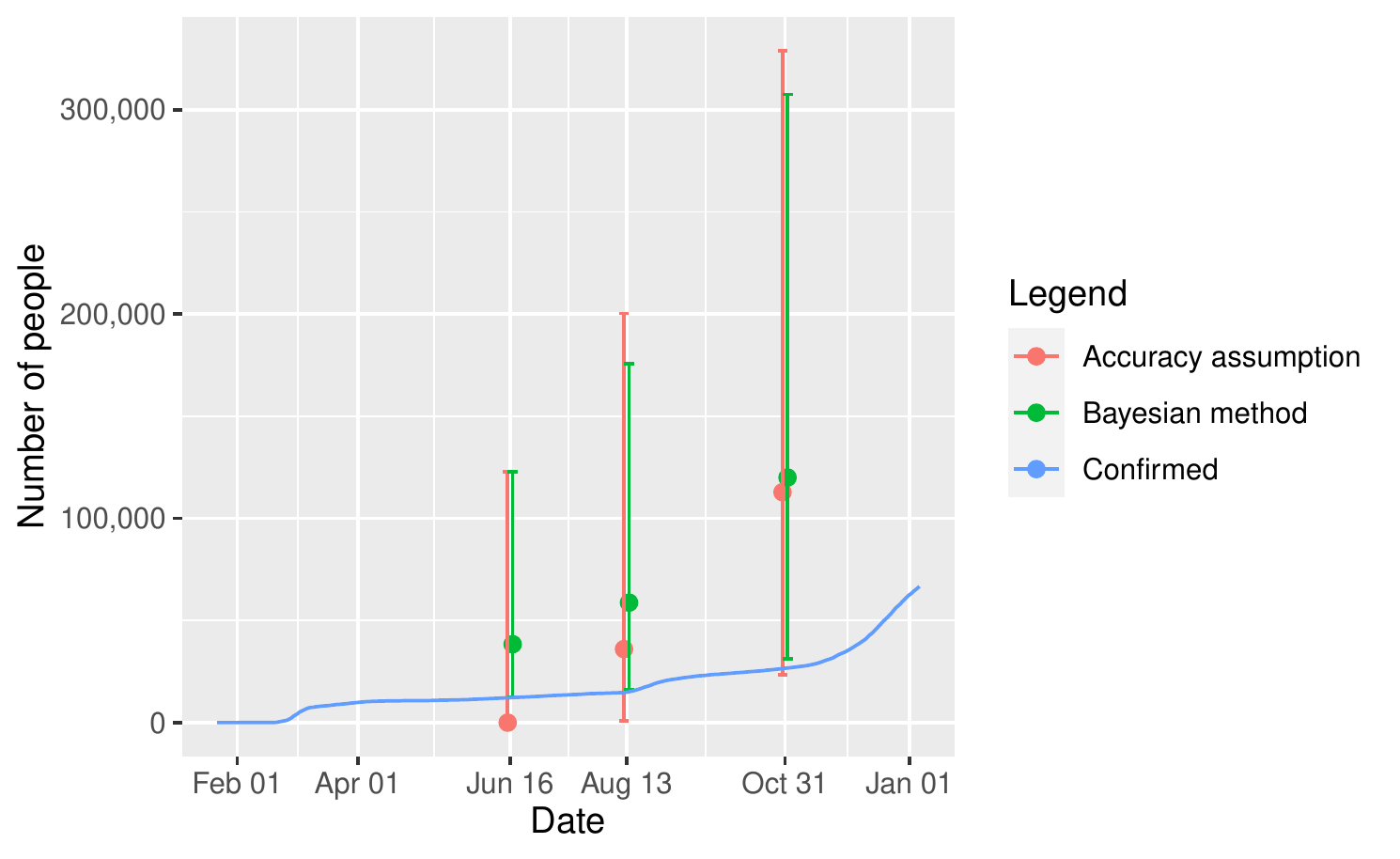}
   	
   	\label{fig:sero}
   \end{figure}
   Figure \ref{fig:sero} shows that the lower bounds of interval estimation by the Bayesian method are larger than the number of confirmed cases as expected, but the other does not satisfy the inequality condition.
   Each upper bound of the interval estimations by the Bayesian method is smaller than the corresponding one obtained under the accuracy assumption.
   Under the inaccuracy assumption, the Bayesian method considers that test-positive cases may include false-negative cases, which is critical when the test-positive number is small enough. Thus, the Bayesian method makes the upper bounds shrink.	
	
	\se{Discussion}
	In this article, we have proposed a Bayesian method with informative prior, which uses the clinical evaluation results of the plaque reduction neutralization test for analyzing data on the seroprevalence surveys in South Korea. We have compared the method with the frequentist's method under the inaccuracy assumption and the statistical analysis under the accuracy assumption. The main advantages of the proposed method are two. First, the method allows the constrained parameter space, which has an obvious lower bound as the proportion of the cumulative confirmed cases. Second, when we consider the inaccuracy assumption, the method can provide a practically corrected estimate contrary to the frequentist's method.

However, this study has a limitation. Each set of samples in the seroprevalence survey does not cover all the regions in South Korea.  
In the first survey announced on the 9th of July, the survey samples do not include those from the populations of several major cities such as Daegu, Daejeon, and Sejong. Daegu particularly was the city of the first mass outbreak in South Korea. The other surveys also do not cover all the cities. The second survey samples do not include those from Ulsan, Busan, Jeonnam, and Jeju, and for the third survey, Gwangju and Jeju are not covered.

%According to the statistical analysis, 
%By comparing the cumulative confirmed cases, estimation of seroprevalence can be used as a measure for the assessment of South Korea's epidemiological investigation, which can be assessed by a comparative study with other countries.
%
%
%Our study has some limitations. 
%First, the survey sample by KDCA lacks representativeness. According to the report by KDCA \citep{kcdc}, the sample does not include cases from the populations of some regions. 
%For example, in the first survey, the sample does not contain cases from large cities such as Daegu, Daejeon, and Sejong, of which Daegu was the city of the first mass outbreak in Korea.  

%% Acknowledgments
\section*{Acknowledgements}
%Seongil Jo was supported by Basic Science Research Program through the National Research Foundation of Korea (NRF) grant funded by the Korea government (MIST) (no. 2020R1C1C1A01013338).
Seongil Jo was supported by INHA UNIVERSITY Research Grant, and  Jaeyong Lee was supported by the National Research Foundation of Korea (NRF) grant funded by the Korea government(MSIT) (No. 2018R1A2A3074973)

	\bibliographystyle{dcu}
	\bibliography{covid}
	
\end{document}